%% file: main.tex
\documentclass[submission,copyright,creativecommons]{eptcs}

\usepackage[final]{microtype}

\usepackage{graphicx}
\usepackage{tikz}
\usetikzlibrary{positioning,arrows}
\usepackage{paralist}

\usepackage{amsmath,amsfonts,amssymb,amsthm}
\usepackage{mathtools}

\DeclareMathOperator{\NP}{\mathbf{NP}}
\DeclareMathOperator{\PSPACE}{\mathbf{PSPACE}}
\DeclareMathOperator{\EXPSPACE}{\mathbf{EXPSPACE}}

\DeclareMathAlphabet{\mathpzc}{OT1}{pzc}{m}{it}

\newcommand{\Beg}{\textit{beg}}
\newcommand{\End}{\textit{end}}

\newcommand{\tpl}[1]{(#1)}

\newcommand{\Nat}{\mathbb{N}}
\newcommand{\true}{\top}

\theoremstyle{plain}
\newtheorem{proposition}{Proposition}
\newtheorem{theorem}[proposition]{Theorem}

\theoremstyle{definition}
\newtheorem{definition}[proposition]{Definition}

\theoremstyle{remark}
\newtheorem{claim}[proposition]{Claim}

\newcommand{\details}[1]{{}}%

\newcommand{\Inst}{{\textit{L}}}
\newcommand{\halt}{{\textit{halt}}}
\newcommand{\init}{{\textit{init}}}
\newcommand{\main}{{\textit{main}}}

\newcommand{\inc}{{\mathsf{inc}}}
\newcommand{\dec}{{\mathsf{dec}}}
\newcommand{\zero}{{\mathsf{zero}}}
\newcommand{\instr}{{\textit{op}}}

\newcommand{\From}{{\textit{from}}}
\newcommand{\To}{{\textit{to}}}
\newcommand{\cont}{{\textit{sec}}}

\newcommand{\start}{\mathsf{s}}
\newcommand{\Ending}{\mathsf{e}}

\newcommand{\startTime}{\mathsf{s}}

\newcommand{\der}[1]{\ensuremath{\;\;{\mathop{{ %
            \longrightarrow}}\limits^{{#1}}}\!}\;\;} %

\newcommand{\TA}{\text{\sffamily TA}}
\newcommand{\MTL}{\text{\sffamily MTL}}
\newcommand{\TPTL}{\text{\sffamily TPTL}}

\newcommand{\RealP}{{\mathbb{R}_+}}

\newcommand{\Intv}{{\mathit{Intv}}}

\title{Undecidability of future timeline-based planning \\ over dense  temporal domains}

\author{Laura Bozzelli \qquad Adriano Peron
\institute{University of Napoli ``Federico II'', Napoli, Italy}
\email{lr.bozzelli@gmail.com \qquad adrperon@unina.it}
\and
Alberto Molinari \qquad Angelo Montanari
\institute{University of Udine, Udine, Italy}
\email{molinari.alberto@gmail.com \qquad angelo.montanari@uniud.it}
}

\def\titlerunning{Complexity of timeline-based planning over dense domains}
\def\authorrunning{L. Bozzelli, A. Molinari, A. Montanari \& A. Peron}

\begin{document}

\maketitle

\begin{abstract}
Planning is one of the most studied problems in computer science. In this paper,
we consider the timeline-based approach, where the domain  is modeled by  a set of independent, but interacting, components, identified by a set of state variables, whose behavior over time (timelines) is governed by a set of temporal constraints (synchronization rules).
Timeline-based planning in the dense-time setting has been recently shown to be undecidable in the general case, and
undecidability relies on the high expressiveness of the trigger synchronization rules.
In this paper, we strengthen the previous negative result by showing that undecidability already holds  under the \emph{future semantics} of the
trigger rules which limits the comparison to temporal contexts in the future with
respect to the trigger.
\end{abstract}

\input{intro.tex}

\input{timelines.tex}

\input{Undecidability.tex}

\bibliographystyle{eptcs}
\bibliography{bib}
\newpage

\end{document}

%% file: intro.tex
\section{Introduction}

\emph{Timeline-based planning} (TP for short) represents a promising approach for real-time temporal planning and reasoning about execution under
uncertainty~\cite{CimattiMR13,MayerOU14,MayerOU16,CestaFFOT09,CestaFFOT10b,CialdeaMayerO15}. Compared to classical  action-based temporal planning~\cite{FoxLong03,Rintanen07},
TP adopts a more declarative paradigm 
which is focused on the constraints that sequences of actions have to fulfill to reach a fixed goal.
In TP, the planning domain is modeled as a set of independent, but interacting, components, each one identified by a 
\emph{state variable}. The temporal behaviour of a single state variable (component) is  described by a sequence of \emph{tokens} (\emph{timeline})
where each token specifies a value of variable (state) and the period of time during which the variable assumes that value.
The overall temporal behaviour (set of timelines) is constrained by a set of \emph{synchronization rules} which specify quantitative temporal requirements
between the time events (start-time and end-time) of distinct tokens.
Synchronization rules have a very simple format: either \emph{trigger rules} expressing invariants and response properties (for each token with a fixed state,  called \emph{trigger}, there exist  tokens satisfying some mutual temporal relations) or \emph{trigger-less rules} expressing goals (there exist tokens satisfying some mutual temporal relations). Note that the way in which timing requirements are specified  in the synchronization rules  corresponds to the ``freeze" mechanism in the
well-known timed temporal logic $\TPTL$~\cite{AlurH94} which uses the freeze quantifier to bind a variable to a specific temporal context (a token in the TP setting).

TP has been successfully exploited in a number of application domains, including space missions, constraint solving, and activity scheduling (see, e.g.,~\cite{Muscettola94,JonssonMMRS00,FrankJ03,CestaCFOP07,aspen2010,barreiro2012europa}). A systematic study of expressiveness and complexity issues for TP has been undertaken only very recently
both in the discrete-time and dense-time settings~\cite{GiganteMCO16,GiganteMCO17,BozzelliMMP18a,BozzelliMMP18b}.
In the discrete-time context, the TP problem is $\EXPSPACE$-complete, and is expressive enough to capture action-based temporal planning  
(see \cite{GiganteMCO16,GiganteMCO17}).

On the other hand, despite the simple format of synchronization rules, the shift to a dense-time domain dramatically increases expressiveness, depicting a scenario which resembles that of the well-known timed linear temporal logics \MTL\ and \TPTL\ (under a pointwise semantics) 
which are undecidable  in the general setting~\cite{AlurH94,OuaknineW06}. 
 In fact the TP problem is undecidable in the general case~\cite{BozzelliMMP18a}, and undecidability relies on the high expressiveness of the trigger rules
  (by restricting the formalism to only trigger-less rules  the problem is just $\NP$-complete~\cite{BozzelliMMP18a}).
Decidability can be recovered by suitable (syntactic/semantic) restrictions on the trigger rules. In particular, in~\cite{BozzelliMMP18b}, two  restrictions are considered:
(i) the first one limits the comparison
to tokens whose start times follow the trigger start time (\emph{future semantics of trigger rules}), and (ii) the second one is syntactical
and imposes that
 a non-trigger token can be referenced at most once in the timed constraints of a trigger rule (\emph{simple trigger rules}).
 Note that the second restriction avoids comparisons of multiple token time-events
with a non-trigger reference time-event.
  Under the previous two restrictions, the TP problem is decidable with a non-primitive recursive complexity~\cite{BozzelliMMP18b} and can be solved by a reduction to model checking of Timed Automata(\TA)~\cite{AlurD94} against \MTL\ over \emph{finite} timed words, the latter being a known decidable problem~\cite{OuaknineW07}.
As in the case of \MTL~\cite{Alur:1996}, better complexity results, i.e. $\EXPSPACE$-completeness (resp., $\PSPACE$-completeness) can be obtained by restricting also the type of \emph{intervals} used in the simple trigger rules in order to compare tokens: non-singular intervals 
(resp., intervals unbounded or starting from $0$).  \vspace{0.1cm}

In this paper, we show that both the considered restrictions on the trigger rules are necessary to recovery decidability. The undecidability of the TP problem with simple trigger rules has been already established in~\cite{BozzelliMMP18a}. Here, we prove undecidability of the  TP problem with arbitrary trigger rules under the future semantics.

%% file: timelines.tex
\section{Preliminaries}\label{sec:preliminaries}

Let $\Nat$ be the set of natural numbers, $\RealP$ be the set of non-negative real numbers, and $\Intv$ be the set of intervals in $\RealP$ whose endpoints are in $\Nat\cup\{\infty\}$. Moreover, let us denote by $\Intv_{(0,\infty)}$ the set of intervals $I\in \Intv$ such that
  either $I$ is unbounded, or $I$
  is left-closed with left endpoint $0$. Such intervals $I$ can be replaced by expressions of the form $\sim n$ for some $n\in\Nat$
  and $\sim\in\{<,\leq,>,\geq\}$.
Let $w$ be a finite word over some alphabet. By $|w|$ we denote the length of $w$. For all  $0\leq i<|w|$,  $w(i)$ is
the $i$-th letter of $w$.

\subsection{The TP Problem}\label{sec:timelines}

In this section, we recall 
the TP framework as presented in \cite{MayerOU16,GiganteMCO16}.
In TP, domain knowledge is encoded by a set of state variables, whose behaviour over time is described by transition functions and synchronization rules.


\begin{definition}
  \label{def:statevar}
  A \emph{state variable} $x$ is a triple $x= (V_x,T_x,D_x)$, where $V_x$ is the \emph{finite domain} of the variable $x$, $T_x:V_x\to 2^{V_x}$ is the \emph{value transition function}, which maps
        each $v\in V_x$ to the (possibly empty) set of successor values, and $D_x:V_x\to \Intv$ is the \emph{constraint function} that maps each $v\in V_x$
        to an interval. 
\end{definition}

A \emph{token} for a variable $x$ is a pair $(v,d)$ consisting of a value $v\in V_x$ and a duration $d\in \RealP$
such that $d\in D_x(v)$. Intuitively, a token for $x$ represents an interval of time where the state variable $x$ takes value $v$.
The behavior of the state variable $x$ is specified by means of \emph{timelines} which are non-empty sequences of tokens
$\pi = (v_0,d_0)\ldots  (v_n,d_n)$  consistent with the value transition function $T_x$, that is, such that
$v_{i+1}\in T_x(v_i)$ for all $0\leq i<n$. The \emph{start time} $\start(\pi,i)$ and the \emph{end time} $\Ending(\pi,i)$ of the $i$-th token ($0\leq i\leq n$) of the timeline
$\pi$  are defined as follows: $\Ending(\pi,i)=\displaystyle{\sum_{h=0}^{i}} d_h$ and  $\start(\pi,i)=0$ if $i=0$, and $\start(\pi,i)=\displaystyle{\sum_{h=0}^{i-1}} d_h$ otherwise.
See Figure~\ref{fig:timelineEx} for an example.
\begin{figure}
    \centering
    \includegraphics{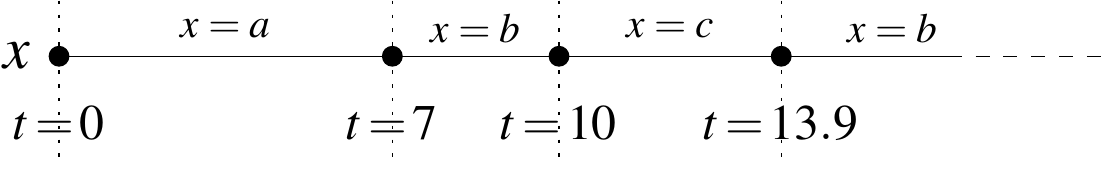}
    \caption{An example of timeline $(a,7)(b,3)(c,3.9)\cdots$ for the state variable $x= (V_x,T_x,D_x)$, where $V_x=\{a,b,c,\ldots\}$, $b\in T_x(a)$, $c\in T_x(b)$, $b\in T_x(c)$\dots and $D_x(a)=[5,8]$, $D_x(b)=[1,4]$, $D_x(c)=[2,\infty[$\dots}
    \label{fig:timelineEx}
\end{figure}




%


Given a finite set $SV$ of state variables, a \emph{multi-timeline} of $SV$ is a mapping $\Pi$ assigning to
each state variable $x\in SV$ a timeline for $x$.
Multi-timelines of $SV$ can be constrained by a set of \emph{synchronization
rules}, which relate tokens, possibly belonging to different timelines, through
temporal constraints on the start/end-times of tokens (time-point constraints) and on the difference
between start/end-times of tokens (interval constraints). The synchronization rules exploit
an alphabet $\Sigma$ of token names to refer to the tokens along a multi-timeline, and are based on the notions of
\emph{atom} and \emph{existential statement}.

%

\begin{definition}
  \label{def:timelines:atom}
  An \emph{atom} is either a clause of the form $o_1\leq^{e_1,e_2}_{I} o_2$
  (\emph{interval atom}), or of the forms $o_1\leq^{e_1}_{I} n$ or  $n\leq^{e_1}_{I}
  o_1$ (\emph{time-point atom}), where $o_1,o_2\in\Sigma$, $I\in\Intv$, $n\in\Nat$, and $e_1,e_2\in\{\start,\Ending\}$.
\end{definition}

An atom $\rho$ is evaluated with respect to a \emph{$\Sigma$-assignment $\lambda_\Pi$ for a given multi-timeline $\Pi$}
which is a mapping assigning to each token name $o\in \Sigma$ a pair $\lambda_\Pi(o)=(\pi,i)$ such that $\pi$ is a timeline of $\Pi$ and $0\leq i<|\pi|$ is a position along $\pi$ (intuitively,
$(\pi,i)$ represents the token of $\Pi$ referenced by the name $o$).
An interval atom $o_1\leq^{e_1,e_2}_{I} o_2$  \emph{is satisfied by  $\lambda_\Pi$} if $e_2(\lambda_\Pi(o_2))-e_1(\lambda_\Pi(o_1))\in I$.
A point atom $o\leq^{e}_{I} n$  (resp., $n\leq^{e}_{I}o$)   \emph{is satisfied by  $\lambda_\Pi$} if $n-e(\lambda_\Pi(o))\in I$ (resp., $e(\lambda_\Pi(o))-n\in I$).


\begin{definition}
 An \emph{existential statement} $\mathcal{E}$ for a finite set $SV$ of state variables is a statement of the form:
\[
\mathcal{E}:=  \exists o_1[x_1=v_1]\cdots \exists o_n[x_n=v_n].\mathcal{C}
\]
  where $\mathcal{C}$ 
  is a conjunction of atoms,
  $o_i\in\Sigma$, $x_i\in SV$, and $v_i\in V_{x_i}$ for each
  $i=1,\ldots,n$. The elements $o_i[x_i=v_i]$ are called
  \emph{quantifiers}. A token name used in $\mathcal{C}$, but not occurring in any
  quantifier, is said to be \emph{free}. Given a $\Sigma$-assignment $\lambda_\Pi$ for a multi-timeline $\Pi$ of $SV$,
  we say that \emph{$\lambda_\Pi$ is consistent with the existential statement $\mathcal{E}$} if for each quantified token name $o_i$,
   $\lambda_\Pi(o_i)=(\pi,h)$ where $\pi=\Pi(x_i)$ and the $h$-th token of $\pi$ has value $v_i$. A multi-timeline $\Pi$ of $SV$ \emph{satisfies} $\mathcal{E}$
   if there exists a $\Sigma$-assignment $\lambda_\Pi$ for $\Pi$ consistent with $\mathcal{E}$ such that each atom in $\mathcal{C}$ is satisfied by
   $\lambda_\Pi$.
\end{definition}

\begin{definition}
  A \emph{synchronization rule} $\mathcal{R}$ for a finite set $SV$ of state variables is a rule of one of the forms
  \[
  o_0[x_0=v_0] \to \mathcal{E}_1\lor \mathcal{E}_2\lor \ldots \lor \mathcal{E}_k, \quad
          \true \to \mathcal{E}_1\lor \mathcal{E}_2\lor \ldots \lor \mathcal{E}_k,
  \]
  where $o_0\in\Sigma$, $x_0\in SV$, $v_0\in V_{x_0}$, and $\mathcal{E}_1, \ldots, \mathcal{E}_k$
  are \emph{existential statements}. 
  In rules of the first
  form (\emph{trigger rules}), the quantifier $o_0[x_0=v_0]$ is called \emph{trigger}, and we require that only $o_0$ may appear free in $\mathcal{E}_i$ (for $i=1,\ldots,n$). In rules of the second form (\emph{trigger-less rules}), we require
  that no token name appears free.
  \newline
\end{definition}

Intuitively, a  trigger $o_0[x_0=v_0]$ acts as a universal quantifier, which
states that \emph{for all} the tokens of the timeline for
the state variable $x_0$, where the variable $x_0$ takes the
value $v_0$, at least one of the existential statements $\mathcal{E}_i$ must be true. 
Trigger-less rules simply assert the
satisfaction of some existential statement. 
 The semantics of synchronization rules is formally defined as follows.

\begin{definition}\label{def:semanticsRules}
Let $\Pi$ be a multi-timeline of a set $SV$ of state variables.
Given a \emph{trigger-less rule} $\mathcal{R}$ of $SV$, \emph{$\Pi$ satisfies $\mathcal{R}$} if $\Pi$ satisfies some existential statement of $\mathcal{R}$.
 Given a \emph{trigger rule} $\mathcal{R}$ of $SV$ with trigger $o_0[x_0=v_0]$, \emph{$\Pi$ satisfies   $\mathcal{R}$} if for every position $i$ of the
 timeline $\Pi(x_0)$ for $x_0$ such that $\Pi(x_0)=(v_0,d)$, there is an existential statement $\mathcal{E}$ of $\mathcal{R}$  and a $\Sigma$-assignment
 $\lambda_\Pi$ for $\Pi$ which is consistent with $\mathcal{E}$ such that $\lambda_\Pi(o_0)= (\Pi(x_0),i)$ and $\lambda_\Pi$ satisfies all the atoms of $\mathcal{E}$.

\end{definition}

In the paper, we focus on a stronger notion of satisfaction of trigger rules, called \emph{ satisfaction under the future semantics}. It requires that all the non-trigger selected tokens
do not start \emph{strictly before} the start-time of the trigger token.

\begin{definition}
  A multi-timeline $\Pi$ of $SV$  \emph{satisfies under the future semantics} a trigger rule $\mathcal{R}= o_0[x_0=v_0] \to \mathcal{E}_1\vee \mathcal{E}_2\vee
  \ldots \vee \mathcal{E}_k$   if $\Pi$ satisfies the trigger rule obtained from
  $\mathcal{R}$ by replacing each existential statement $\mathcal{E}_i=\exists o_1[x_1=v_1]\cdots \exists o_n[x_n=v_n].\mathcal{C}$
  with  $\exists o_1[x_1=v_1]\cdots \exists o_n[x_n=v_n].\mathcal{C}\wedge  \bigwedge_{i=1}^{n} o_0\leq^{\start,\start}_{[0,+\infty[} o_i$.
\end{definition}


A TP domain  $P=(SV,R)$ is specified by a finite set $SV$ of state variables and
a finite set $R$ of synchronization rules modeling their admissible behaviors.
Trigger-less rules can be used to express initial conditions
and the goals of the problem, while trigger rules are useful to specify invariants and response requirements. A \emph{plan of $P$} is a  multi-timeline of $SV$ satisfying all the rules in $R$. A \emph{ future plan of $P$} is defined in a similar way, but we require that the fulfillment of the trigger rules is under the future semantics.
We are interested in the \emph{Future TP problem} consisting in checking
for a given TP domain $P=(SV,R)$, the existence of a future plan for $P$. 

\details{
Table~\ref{tab:complex} summarizes all the decidability and complexity results described in the following.
We consider mixes of restrictions of the TP problem involving trigger rules with future semantics, simple trigger rules, and  intervals in atoms of trigger rules which are non-singular or in $\Intv_{(0,\infty)}$.

\begin{table}[b]
    \centering
    \resizebox{\linewidth}{!}{
    \begin{tabular}{r|c|c}
    	& TP problem & Future TP problem \\
    	\hline
    	Unrestricted & Undecidable & (Decidable?) Non-primitive recursive-hard \\
    	\hline
    	Simple trigger rules & Undecidable & Decidable (non-primitive recursive) \\
    	\hline
    	Simple trigger rules, non-singular intervals & ? & $\EXPSPACE$-complete \\
    	\hline
    	Simple trigger rules, intervals in $\Intv_{(0,\infty)}$ & ? & $\Psp$-complete \\
    	\hline
    	Trigger-less rules only & $\NP$-complete & // \\
    \end{tabular} }
    \caption{Decidability and complexity of restrictions of the TP problem.}
    \label{tab:complex}
\end{table}}

%% file: Undecidability.tex
\section{Undecidability of the future TP problem}\label{sec:undecidability}

In this section, we establish the following result.
\begin{theorem}\label{theorem:undecidability}
Future TP with \emph{one state variable} is undecidable  even if  the intervals  are in $\Intv_{(0,\infty)}$.
\end{theorem}

Theorem~\ref{theorem:undecidability} is proved by a polynomial-time reduction
from the \emph{halting problem for Minsky $2$-counter machines}~\cite{Minsky67}.
Such a machine is a tuple $M = \tpl{Q,q_\init,q_\halt,  \Delta}$,
where  $Q$ is a finite set of (control) locations, $q_\init\in Q$ is the initial location,   $q_\halt\in Q$ is the halting location, and
   $\Delta \subseteq Q\times \Inst \times Q$ is a transition relation over the instruction set $\Inst= \{\inc,\dec,\zero\}\times \{1,2\}$.

We adopt the following notational conventions.
 For an instruction
$\instr=(\_\,,c)\in \Inst$, let $c(\instr):=c\in\{1,2\}$ be the \emph{counter}
associated
with $\instr$.
For a transition $\delta\in \Delta$ of the form $\delta=(q,\instr,q')$, we
define $\From(\delta):= q$, $\instr(\delta):=\instr$, $c(\delta):= c(\instr)$,
and $\To(\delta):= q'$.
Without loss of generality, we make these assumptions:
\begin{compactitem}
  \item for each transition $\delta\in \Delta$, $\From(\delta)\neq q_\halt$ and $\To(\delta )\neq q_\init$, and
    \item there is exactly one transition in $\Delta$, denoted $\delta_\init$, having as source the initial location $q_\init$.
\end{compactitem}

An $M$-configuration is a pair $(q,\nu)$ consisting of a location $q\in Q$ and a counter valuation $\nu: \{1,2\}\to \Nat$.
$M$ induces a transition relation, denoted by $\longrightarrow$, over pairs of $M$-configurations  defined as follows.
 For configurations $(q,\nu)$ and $(q',\nu')$, $(q,\nu) \der{} (q',\nu')$ if for some instruction $\instr\in \Inst$,  $(q,\instr,q')\in \Delta$ and the following holds, where $c\in \{1,2\}$ is the counter associated with the instruction
 $\instr$:
\begin{inparaenum}[(i)]
  \item    $\nu'(c')= \nu(c')$ if   $c'\neq c$;
  \item  $\nu'(c)= \nu(c) +1$ if $\instr=(\inc,c)$;
  \item $\nu(c)>0$ and $\nu'(c)= \nu(c) -1$ if $\instr=(\dec,c)$; and
   \item  $\nu'(c)= \nu(c)=0$ if $\instr=(\zero,c)$.
\end{inparaenum}

A  computation of $M$ is a non-empty \emph{finite} sequence $C_1,\ldots ,C_k$ of configurations such that $C_i \der{} C_{i+1}$ for all $1\leq i<k$.
$M$ \emph{halts} if there is a computation starting at the \emph{initial}
configuration $(q_\init, \nu_\init)$, where $\nu_\init(1) = \nu_\init(2) = 0$,
and leading to some
halting configuration
$(q_{\halt}, \nu)$. The halting problem is to decide whether a given machine $M$ halts, and it is was proved to be undecidable~\cite{Minsky67}.
We prove the following result, from which Theorem~\ref{theorem:undecidability} directly follows.

\begin{proposition}\label{prop:undecidability}
One can construct (in polynomial time) a TP instance (domain)
$P=(\{x_M\},R_M)$  where the intervals in $P$ are in $\Intv_{(0,\infty)}$
such that $M$ halts \emph{iff} there exists a future plan for $P$.
\end{proposition}
\begin{proof}
First, we define a suitable encoding of a computation of $M$ as the untimed
part of a timeline (i.e., neglecting tokens' durations and accounting only for
their values) for $x_M$. For this,
we exploit the finite set of symbols $V:= V_{\main}\cup V_{\cont}$ corresponding to the finite domain of the state variable $x_M$.
The   set   of \emph{main} values $V_{\main}$ is the set of $M$-transitions, i.e. $V_{\main}= \Delta$.
The set of \emph{secondary} values $V_{\cont}$ is defined as
$ V_\cont := \Delta \times \{1,2\} \times  \{\#,\Beg,\End\}$, where $\#$, $\Beg$, and $\End$ are three special symbols used as markers.
Intuitively, in the encoding of an $M$-computation a main value  keeps track of the transition  used in the current step of the computation, while
the set $V_{\cont}$ is used for encoding counter values.

For $c\in \{1,2\}$, a \emph{$c$-code for the main value $\delta\in\Delta$} is a  finite word $w_c$ over $V_\cont$ of the form
$(\delta,c,\Beg)\cdot (\delta,c,\#)^{h}\cdot (\delta,c,\End)$ for some $h\geq 0$ such that $h=0$ if $\instr(\delta)= (\zero,c)$.  The $c$-code $w_c$ encodes the value for counter $c$
given by $h$ (or equivalently $|w_c|-2$). Note that only the occurrences of the symbols $(\delta,c,\#)$ encode units in the value of counter $c$, while  the symbol $(\delta,c,\Beg)$ (resp., $(\delta,c,\End)$) is only used as left (resp., right) marker in the encoding.

A \emph{configuration-code $w$  for a main value $\delta\in\Delta$} is a finite word over $V$
of the form $w= \delta \cdot w_1 \cdot w_2 $ such that for each counter $c\in \{1,2\}$, $w_c$ is a $c$-code
for the main value $\delta$. The configuration-code $w$ encodes the $M$-configuration $(\From(\delta),\nu)$, where $\nu(c)=|w_c|-2$
for all $c\in \{1,2\}$. Note that if $\instr(\delta)=(\zero,c)$, then $\nu(c)=0$.

A \emph{computation}-code is a non-empty sequence of configuration-codes $\pi= w_{ \delta_1} \cdots w_{ \delta_k}$, where  for all $1\leq i\leq k$, $w_{ \delta_i}$ is a configuration-code with main value $ \delta_i$, and whenever
  $i<k$, it holds that $\To(\delta_i)=\From(\delta_{i+1})$. Note that by our assumptions $\To(\delta_i)\neq q_\halt$ for all $1\leq i<k$, and
  $\delta_j\neq \delta_\init$ for all $1<j\leq k$.
  The computation-code $\pi$ is \emph{initial} if  the first configuration-code
  $w_{\delta_1}$ has the main value $ \delta_\init$ and encodes the initial
  configuration, and it is \emph{halting} if
  for the last  configuration-code $w_{\delta_k}$ in $\pi$, it holds that $\To(\delta_k)=q_\halt$.
For all $1\leq i\leq k$, let $(q_i,\nu_i)$ be the $M$-configuration encoded by the configuration-code $w_{\delta_i}$ and $c_i= c(\delta_i)$.
 The computation-code $\pi$ is \emph{well-formed} if, additionally,   for all $1\leq j< k$,  the following holds:
\begin{compactitem}
  \item $\nu_{j+1}(c)=\nu_j(c)$ if either $c \neq c_j$ or $\instr(\delta_j)= (\zero,c_j)$  (\emph{equality requirement});
\item $\nu_{j+1}(c_j)= \nu_j(c_j)+1$ if $\instr(\delta_j)= (\inc,c_j)$ (\emph{increment requirement});
 \item $\nu_{j+1}(c_j)= \nu_j(c_j)-1$ if $\instr(\delta_j)= (\dec,c_j)$ (\emph{decrement requirement}).
\end{compactitem}

\noindent Clearly,
$M$ halts \emph{iff} there exists an initial and halting well-formed
computation-code.

\paragraph{Definition of $x_M$ and $R_M$.} We now define a state variable $x_M$ and a set $R_M$ of  synchronization rules for $x_M$ with intervals in $\Intv_{(0,\infty)}$ such that the untimed part of every \emph{future plan} of $P=(\{x_M\},R_M)$
is an initial and halting well-formed computation-code. Thus, $M$ halts iff there is a future plan of $P$.

Formally, 
variable $x_M$ is given by $x_M= (V= V_{\main}\cup V_{\cont},T,D)$, where for
each $v\in V$,
$D(v)=]0,\infty[$. Thus, we require that the duration of a  token 
is always greater than zero (\emph{strict time monotonicity}).
The value transition function $T$ of $x_M$ ensures the following property.
\begin{claim}\label{ref:claim}
The untimed parts of the timelines for $x_M$ whose first token has value $\delta_\init$ correspond
 to the prefixes of  initial computation-codes. Moreover, $\delta_\init\notin T(v)$ for all $v\in V$.
\end{claim}

%

 By construction, it is a trivial task to define $T$ so that the previous
 requirement is fulfilled.

 Let $V_\halt=\{\delta\in \Delta\mid \To(\delta)=q_\halt\}$.
 By Claim~\ref{ref:claim} and the assumption that  $\From(\delta)\neq q_\halt$
 for each transition $\delta\in \Delta$, in order to enforce  the
 initialization and
 halting requirements,
 it suffices  to ensure that a timeline has a token with value $ \delta_\init$ and a token with value in $V_\halt$. This is captured by the trigger-less rules
    $
   \true \rightarrow \exists   o[x_M= \delta_\init].  \true
   $ and  $\true \rightarrow \bigvee_{v\in V_\halt} \exists   o[x_M=v].  \true $.

Finally, the crucial well-formedness requirement is captured by the trigger
rules in $R_M$ which express punctual time constraints%
\footnote{Such punctual contrains are expressed by pairs of conjoined atoms
whose intervals are in $\Intv_{(0,\infty)}$.}.
We refer the reader to Figure~\ref{fig:enc}, that gives an intuition on the properties enforced by the rules
we are about to define. In particular, we essentially take advantage of the dense temporal domain to allow
for the encoding of arbitrarily large values of counters in one time units.
\begin{figure}
    \centering
    \includegraphics[width=\textwidth]{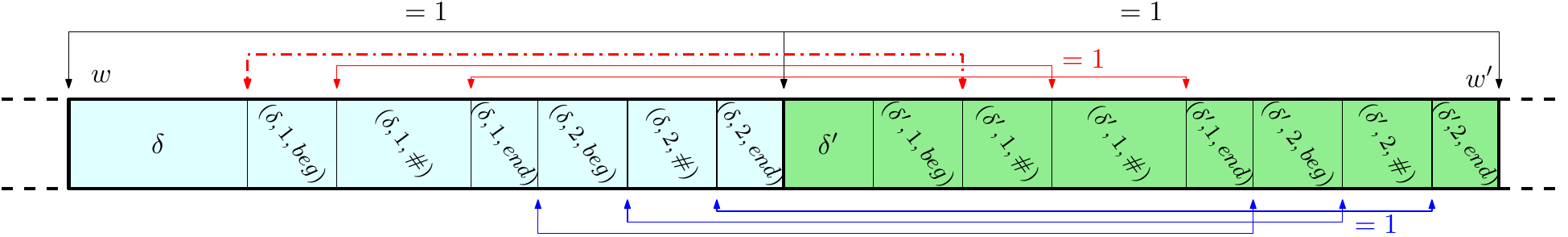}
    \caption{The figure shows two adjacent configuration-codes, $w$ (highlighted in cyan) and $w'$ (in green), the former for $\delta=(q,(\mathsf{inc},1),q')\in\Delta$ and the latter for $\delta'=(q',\dots)\in\Delta$;
    $w$ encodes the $M$-configuration $(q,\nu)$ where $\nu(1)=\nu(2)=1$, and $w'$ the $M$-configuration $(q',\nu')$ where $\nu'(1)=2$ and $\nu'(1)=1$.\newline
    The \lq\lq 1-Time distance between consecutive main values requirement\rq\rq\ (represented by black lines with arrows) forces a token with a main value to be followed, after exactly one time instant, by another token with a main value.\newline
    Since $\instr(\delta)= (\inc,1)$, the value of counter 2 does not change in this computation step, and thus the values for counter 2 encoded by $w$ and $w'$ must be equal. To this aim the \lq\lq equality requirement\rq\rq\ (represented by blue lines with arrows) sets a one-to-one correspondence between pairs of tokens associated with counter 2 in $w$ and $w'$ (more precisely, a token $tk$ with value $(\delta,2,\_)$ in $w$ is followed by a token
    $tk'$  with value $(\delta',2,\_)$ in $w'$ such that  $\start(tk')-\start(tk)=1$ and
   $\Ending(tk')-\Ending(tk)=1$).\newline
   Finally, the \lq\lq increment requirement\rq\rq\ (red lines) performs the increment of counter 1 by doing something analogous to the previous case, but with a difference: the token $tk'$ with value $(\delta',1,\#)$ is in $w'$ in the place where the token $tk$ with value $(\delta,1,\Beg)$ was in $w$ (i.e., $\start(tk')-\start(tk)=1$ and
   $\Ending(tk')-\Ending(tk)=1$). The token $tk''$ with value $(\delta',1,\Beg)$ is \lq\lq anticipated\rq\rq , in such a way that $\Ending(tk'')-\start(tk)=1$ (this is denoted by the dashed red line): the token with main value $\delta'$ in $w'$ has a shorter duration than that with value $\delta$ in $w$, leaving space for $tk''$, so as to represent the unit added by $\delta$ to counter 1. Clearly density of the time domain plays a fundamental role here.
    }
    \label{fig:enc}
\end{figure}

\paragraph{Trigger rules for 1-Time distance between consecutive main values.}
We define non-simple trigger rules requiring that
the overall duration of the sequence of tokens corresponding to a configuration-code  amounts exactly to one time units.
By Claim~\ref{ref:claim}, strict time monotonicity, and the halting
requirement, it suffices to ensure that each token $tk$ having a  main value in
$V_\main \setminus V_\halt$ is eventually followed by a token $tk'$  such that
$tk'$ has a  main value and $\startTime(tk')-\startTime(tk)=1$ (this
denotes---with a little abuse of notation---that the difference of start times
is exactly $1$). To this aim, for each $v\in V_\main \setminus V_\halt$, we
write the non-simple trigger rule with intervals in
$\Intv_{(0,\infty)}$:
\[
o[x_M=v] \rightarrow \bigvee_{u\in V_\main}  \exists  o'[x_M= u] .\,
 o\leq^{\start,\start}_{[1,+\infty[} o' \,\wedge\,
 o\leq^{\start,\start}_{[0,1]} o' .
\]

\paragraph{Trigger rules for the equality requirement.} In order to ensure the
equality requirement, we exploit the fact that the end time of a token along a
timeline
corresponds to the start time of the next token (if any). Let $V_\cont^{=}$ be the set of
 secondary states $(\delta,c,t)\in V_\cont$ such that $\To(\delta)\neq q_\halt$, and either $c\neq c(\delta)$ or $\instr(\delta)= (\zero,c)$. Moreover, for a counter $c\in \{1,2\}$ and a tag $t\in \{\Beg,\#,\End\}$, let $V_c^{t}\subseteq V_\cont$ be the set of secondary states given
    by $\Delta\times \{c\} \times \{t\}$.  We require the following:
 \begin{compactenum}
 \item [(*)] each token $tk$  with
    a $(V_{c}^{t}\cap V_\cont^{=})$-value is eventually followed by a token
    $tk'$ with a $V_{c}^{t}$-value such that  $\start(tk')-\start(tk)=1$ (i.e.,
    the difference of start times is exactly $1$). Moreover, if $t\neq \End$,
    then
   $\Ending(tk')-\Ending(tk)=1$  (i.e., the difference of end times is exactly $1$).
  \end{compactenum}
 Condition~(*) is captured by the following non-simple trigger rules with
 intervals in $\Intv_{(0,\infty)}$:
\begin{compactitem}
	\item  for each $v\in V^{t}_c\cap V_\cont^{=}$ and $t\neq \End$,
	\[
	o[x_M=v] \rightarrow \bigvee_{u\in V^{t}_c} \exists  o'[x_M= u].\,
	o\leq^{\start,\start}_{[1,+\infty[} o' \,\wedge\,
	o\leq^{\start,\start}_{[0,1]} o' \, \wedge\,
	o\leq^{\Ending,\Ending}_{[1,+\infty[} o' \,\wedge\,
	o\leq^{\Ending,\Ending}_{[0,1]} o';
	\]
	\item for each $v\in V^{\End}_c\cap V_\cont^{=}$,
	\[
	o[x_M=v] \rightarrow \bigvee_{u\in V^{\End}_c}   \exists  o'[x_M= u].\,
	o\leq^{\start,\start}_{[1,+\infty[} o'   \, \wedge\,
	o\leq^{\start,\start}_{[0,1]} o'.
	\]
\end{compactitem}

We now show that Condition~(*) together with strict time monotonicity and
1-Time distance between consecutive main values ensure the equality
requirement. Let $\pi$ be a timeline of $x_M$ satisfying all the rules defined
so far,  $w_\delta$ and $w_{\delta'}$ two \emph{adjacent} configuration-codes
along
$\pi$ with $w_\delta$ preceding $w_{\delta'}$ (note that $\To(\delta)\neq
q_\halt$), and $c\in\{1,2\}$ a counter such that either $c\neq c(\delta)$ or
$\instr(\delta)= (\zero,c)$.
Let $tk_0\cdots tk_{\ell+1}$ (resp., $tk'_0\cdots tk'_{\ell'+1}$) be the
sequence of tokens associated with the $c$-code of $w_\delta$ (resp.,
$w_{\delta'}$).
We need to show that $\ell=\ell'$. By construction $tk_0$ and $tk'_0$ have value in  $V_c^{\Beg}$, $tk_{\ell+1}$ and $tk'_{\ell'+1}$ have value in $ V_c^{\End}$, and for all $1\leq i\leq \ell$ (resp., $1\leq i'\leq \ell'$), $tk_i$ has value in $V_c^{\#}$ (resp., $tk'_{i'}$ has value in $V_c^{\#}$). Then strict time monotonicity,   1-Time distance between consecutive main values, and Condition~(*) guarantee
the existence of an \emph{injective} mapping $g: \{tk_0,\ldots,tk_{\ell+1}\}
\rightarrow \{tk'_0,\ldots,tk'_{\ell'+1}\}$ such that
$g(tk_0)=tk'_0$, $g(tk_{\ell+1})=tk'_{\ell'+1}$, and for all $0\leq i\leq \ell$, if $g(tk_i)=tk'_j$ (note that $j<\ell'+1$), then
$g(tk_{i+1})=tk'_{j+1}$ (we recall that the end time of a token is equal to the
start time of the next token along a timeline, if any). These properties ensure
that $g$ is \emph{surjective}
as well. Hence, $g$ is a bijection and $\ell'=\ell$.

\paragraph{Trigger rules for the increment requirement.}   Let $V_\cont^{\inc}$
be the set of
 secondary states $(\delta,c,t)\in V_\cont$ such that $\To(\delta)\neq
 q_\halt$  and $\instr(\delta)= (\inc,c)$. By reasoning like in the case of the
 rules ensuring the equality requirement, in order to express the increment
 requirement, it suffices to enforce the following conditions for each counter
 $c\in \{1,2\}$:
 \begin{compactenum}[(i)]
 \item each token $tk$  with
    a $(V_{c}^{\Beg}\cap V_\cont^{\inc})$-value is eventually followed by a
    token $tk'$ with a $V_{c}^{\Beg}$-value such that
    $\Ending(tk')-\start(tk)=1$ (i.e., the difference between the end time of
    token $tk'$ and the start time of token $tk$ is exactly $1$);
 \item for each $t\in \{\Beg,\#\}$,  each token $tk$  with
    a $(V_{c}^{t}\cap V_\cont^{\inc})$-value is eventually followed by a token
    $tk'$   with a $V_{c}^{\#}$-value such that  $\start(tk')-\start(tk)=1$
    and $\Ending(tk')-\Ending(tk)=1$ (i.e., the difference of start times and
    end times is exactly $1$). Observe that
    the token with a $(V_{c}^{\Beg}\cap V_\cont^{\inc})$-value is associated
    with
    a token with
    $V_{c}^{\#}$-value anyway;
  \item each token $tk$  with
    a $(V_{c}^{\End}\cap V_\cont^{\inc})$-value is eventually followed by a
    token $tk'$ with a $V_{c}^{\End}$-value such that
    $\start(tk')-\start(tk)=1$ (i.e., the difference of start times is exactly
    $1$);
  \end{compactenum}
 Intuitively,
 if $w$ and $w'$ are two \emph{adjacent} configuration-codes
 along a timeline of $x_M$, with $w$ preceding $w'$,
 (i) and (ii) force a token $tk'$ with a $V_{c}^{\#}$-value in $w'$ to \lq\lq
 take the place\rq\rq\ of the token $tk$
 with
 $(V_{c}^{\Beg}\cap V_\cont^{\inc})$-value in $w$ (i.e., they have the same
 start and
 end times). Moreover a token with $V_{c}^{\Beg}$-value must immediately
 precede $tk'$ in $w'$.

 These requirements can be expressed by non-simple trigger rules with intervals
 in $\Intv_{(0,\infty)}$ similar to the ones defined for the equality
 requirement.

\paragraph{Trigger rules for the decrement requirement.} For capturing the
decrement requirement, it suffices to enforce the following conditions for each
counter $c\in \{1,2\}$, where  $V_\cont^{\dec}$ denotes the set of
 secondary states $(\delta,c,t)\in V_\cont$ such that $\To(\delta)\neq q_\halt$  and $\instr(\delta)= (\dec,c)$:
 \begin{compactenum}[(i)]
 \item each token $tk$  with
    a $(V_{c}^{\Beg}\cap V_\cont^{\dec})$-value is eventually followed by a
    token $tk'$ with a $V_{c}^{\Beg}$-value such that
    $\start(tk')-\Ending(tk)=1$ (i.e., the difference between the start time of
    token $tk'$ and the end time of token $tk$ is exactly $1$);
 \item each token $tk$  with
    a $(V_{c}^{\#}\cap V_\cont^{\dec})$-value is eventually followed by a token
    $tk'$ with a $V_{c}^{t}$-value where  $t\in \{\Beg,\#\}$    such that
    $\start(tk')-\start(tk)=1$  and $\Ending(tk')-\Ending(tk)=1$ (i.e., the
    difference of start times and end times is exactly $1$).
  \item each token $tk$  with
    a $(V_{c}^{\End}\cap V_\cont^{\dec})$-value is eventually followed by a
    token $tk'$ with a $V_{c}^{\End}$-value such that
    $\start(tk')-\start(tk)=1$ (i.e., the difference of start times is exactly
    $1$);
  \end{compactenum}
Analogously, (i) and (ii) produce an effect which is symmetric w.r.t.\ the case
of increment.

Again, these requirements can be easily expressed by non-simple trigger rules
with intervals in $\Intv_{(0,\infty)}$ as done before for expressing the
equality requirement.

 By construction, the untimed part of a future plan of $P=(\{x_M\},R_M)$ is an initial and halting well-formed computation-code. Vice versa, by exploiting denseness of the temporal domain, the existence of an initial and halting well-formed computation-code implies the existence of a future plan of $P$. This concludes the proof of Proposition~\ref{prop:undecidability}.\qedhere
\end{proof}